\documentclass[11pt,english,letterpaper]{llncs}
\usepackage[T1]{fontenc}
\usepackage[latin9]{inputenc}
\usepackage{amsmath}
\usepackage{todonotes}
\usepackage{fancyhdr}
\usepackage[affil-it]{authblk}
\usepackage[margin=2.5cm]{geometry}
\usepackage{tabularx,caption}
\usepackage[ruled]{algorithm2e}
\usepackage{babel}
\usepackage{hyperref}
\usepackage{wrapfig}
 \usepackage{subcaption}


\newcommand{\ie}{i.\,e.\xspace}
\newcommand{\eg}{e.\,g.\xspace}
\newcommand{\etal}{et al.\xspace}

\newcount\shortyear\newcount\shorthour\newcount\shortminute
\shorthour=\time\divide\shorthour by 60\shortyear=\shorthour
\multiply\shortyear by 60\shortminute=\time\advance\shortminute by
-\shortyear
\shortyear=\year\advance\shortyear by -1900

\def\zeit{\number\shorthour:\ifnum\shortminute<10 0\number\shortminute
\else\number\shortminute\fi}

\def\rightmark{{\small \today, \zeit{}}}
\include{header}

\begin{document}

\author{Elisabetta Bergamini \and Henning Meyerhenke \and Christian L. Staudt}
\institute{Institute of Theoretical Informatics, Karlsruhe Institute of Technology (KIT), Germany \\
  Email: \email{\{elisabetta.bergamini, meyerhenke, christian.staudt\}\,@\,kit.edu}
}

\title{Approximating Betweenness Centrality in Large Evolving Networks}

\date{}

\maketitle

\begin{abstract}
Betweenness centrality ranks the importance of nodes by their participation in \emph{all} shortest paths of the network.
Therefore computing exact betweenness values is impractical in large networks. For static networks,
approximation based on randomly sampled paths has been shown to be significantly faster in practice.
However, for dynamic networks, no approximation algorithm for betweenness centrality
is known that improves on static recomputation. 
We address this deficit by proposing two incremental approximation algorithms (for weighted and unweighted connected graphs) which provide a provable guarantee on the absolute approximation error. 
Processing batches of edge insertions, our algorithms yield significant speedups up to a factor of $10^4$ compared to restarting the approximation.
This is enabled by investing memory to store and efficiently update shortest paths.
As a building block, we also propose an asymptotically faster algorithm for updating the SSSP problem in unweighted graphs.
Our experimental study shows that our algorithms are the first to make in-memory computation of a betweenness ranking practical for million-edge semi-dynamic networks. Moreover, our results show that the accuracy is even better than the theoretical guarantees in terms of absolutes errors and the rank of nodes is well preserved, in particular for those with high betweenness.
\\[0.25ex]

\end{abstract}

% page limit for ALENEX: 10 pages excluding title page and bibliography
%\newpage 

\section{Introduction}
\label{sec:introduction}
% Context
The algorithmic analysis of complex networks has become 
a highly active research area recently. One important task in network analysis is to rank
nodes by their structural importance using \emph{centrality measures}.
%
% Definition and applications
\emph{Betweenness
centrality} (BC) measures the participation of a node
in the shortest paths of the network. 
Let the graph $G$ represent a network with $n$ nodes and $m$ edges.
Naming $\sigma_{st}$ the number of shortest paths from a node $s$ to a node $t$ and $\sigma_{st}(v)$ the number of shortest paths from $s$ to $t$ that go through $v$, the (normalized) BC
of $v$ is defined as~\cite{citeulike:1025135}:
$c_{B}(v)=\frac{1}{n(n-1)}\sum_{s\neq v \neq t}\frac{\sigma_{st}(v)}{\sigma_{st}}$.
 Nodes with high betweenness can be important in routing, spreading processes
and mediation of interactions.
Depending on the context, this can mean, for example, finding the
most influential persons in a social network, the key infrastructure
nodes in the internet, or super spreaders of a disease.

% State of the art
The fastest existing method for computing BC is due to Brandes~\cite{Brandes01betweennessCentrality}. It
requires $\Theta(nm)$ operations for unweighted graphs and $\Theta(nm+n^{2}\log n)$
for graphs with positive edge weights. 
This time complexity is prohibitive for large-scale graphs with millions of edges, though.
Recent years have seen the publication of several ap proximation algorithms
that aim to reduce the computational effort, while finding BC values
that are as close as possible to the exact
ones. Good results have been obtained in this regard; in particular,
a recent algorithm by Riondato and Kornaropoulos (\textsf{RK})~\cite{DBLP:conf/wsdm/RiondatoK14}  gives probabilistic guarantees on the quality of the approximation, and we build our algorithms on this method.

\paragraph{Motivation.}
Large graphs of interest, such as the Web and social networks,
evolve continuously. Thus, in addition to the identification of important 
nodes in a static network, an issue of great interest is the evolution of centrality 
values in dynamically changing networks.  
So far, there are no approximation
algorithms that efficiently update BC scores rather than recomputing them from scratch.
A few methods have been proposed to update the BC values after a graph modification, which for some of the
algorithms can only be one edge insertion and for others can also be
one edge deletion. However, all of these approaches are exact and have
a worst-case time complexity which is the same as Brandes's algorithm (\textsf{BA})
\cite{Brandes01betweennessCentrality} and a memory footprint of at least $O(n^2)$. 

\paragraph{Contribution.}
In this paper, we present the first approximation algorithms for BC in evolving networks.
They are the first incremental algorithms (\ie they update the betweenness values of all nodes
in response to edge insertions or weight decrease operations in the graph) that can actually be used in large streaming graphs.
Our two algorithms \textsf{IA} (for \textit{incremental approximation}) and \textsf{IAW} (for \textsf{IA} \emph{weighted})
work for 
%They are two semi-dynamic
%incremental algorithms, for 
unweighted and weighted networks, respectively. 
%They support
%only edge insertions (and weight decrease operations) to a connected graph. 
Even though only edge insertions and weight decreases are supported as dynamic operations,
we do not consider this a major limitation since many real-world dynamic networks evolve only 
this way %by incremental edge insertions
and do not shrink.
The algorithms we propose are also the first that can handle a batch
of edge insertions (or weight decrease
operations) at once. After each batch of updates, we guarantee that
the approximated BC values differ by at most $\epsilon$
from the exact values with probability at least $1-\delta$, where
$\epsilon$ and $\delta$ can be arbitrarily small constants. 
Running time and memory required depend on how tightly the error should be bounded.
As part of our algorithm \textsf{IA}, we also propose a new algorithm with lower time complexity for updating single-source shortest paths in unweighted graphs after a batch of edge insertions.

Our experimental study shows that our algorithms are the first to make in-memory computation of a betweenness ranking practical for large dynamic networks.
For comparison, we create and evaluate the first implementation of the dynamic BC algorithm by Nasre \etal~\cite{DBLP:journals/corr/NasrePR13} and observe its lack of speedup and scalability.
We achieve a much improved scaling behavior,
enabling us to update approximate betweenness scores in a network with 16 million edges in a few seconds on typical workstation hardware,
where previously proposed dynamic algorithms for betweenness would fail by their memory requirements alone.
More generally, processing batches of edge insertions, our algorithms yield significant speedups (up to factor $10^4$) compared to restarting the approximation.
Regarding accuracy, our experiments show that the estimated absolute errors are always lower than the guaranteed ones.
Also the rank of nodes is well preserved, in particular for those with high betweenness.

\section{Related work}
\label{sec:related_work}

\paragraph*{Static BC algorithms - exact and approximation.}

The fastest existing method for the exact BC computation, \textsf{BA},
requires $\Theta(nm)$ operations for unweighted graphs
and $\Theta(nm+n^{2}\log n)$ for graphs with positive edge weights~\cite{Brandes01betweennessCentrality}.
The algorithm computes for every node $s\in V$
a slightly-modified version of a single-source shortest-path tree
(SSSP tree), producing for each $s$ the directed acyclic
graph (DAG) of \emph{all} shortest paths starting at $s$. Exploiting
the information contained in the DAGs, the algorithm computes the \textit{dependency}
 $\delta_{s}(v)$ for each node $v$, that is
the sum over all nodes $t$ of the fraction of shortest paths between
$s$ and $t$ that $v$ is internal to. The betweenness
of each node $v$ is simply the \emph{sum} over all sources $s\in V$
of the dependencies $\delta_{s}(v)$. Therefore, we can
see the dependency $\delta_{s}(v)$ as a \emph{contribution}
that $s$ gives to the computation of $c_{B}(v)$.
Based on this concept, some algorithms for an
\emph{approximation} of BC have been developed.
Brandes and Pich~\cite{DBLP:journals/ijbc/BrandesP07} propose to
approximate $c_{B}(v)$ by extrapolating it from the contributions
of a \emph{subset} of source nodes, also
called \emph{pivots}. Selecting the pivots uniformly at random, the
approximation can be proven to be an unbiased
estimator for $c_{B}(v)$ (i.e.\, its expectation is equal to $c_{B}(v)$).
In a subsequent work, Geisberger \etal~\cite{DBLP:conf/alenex/GeisbergerSS08}
notice that this can produce an overestimation of BC scores
of nodes that happen to be close to the sampled pivots. To limit
this bias, they introduce
a \emph{scaling function} which gives less importance to
contributions from pivots that are close to the node. 
%In particular, the authors propose two possibilities: a linear scaling function,
%where the effect of the contribution of a pivot on a node $v$ increases
%linearly with the distance between the pivot and $v$, and a bisection
%scaling function, which considers only contributions of pivots that
%are ``far enough'' from $v$ and ignores contributions from other
%pivots. 
%In contrast to the previous works, which consider the problem of
%estimating the betweenness of \textit{all} nodes in the graph, 
Bader \etal~\cite{DBLP:conf/waw/BaderKMM07} approximate
the BC of a specific node only, based on 
an adaptive sampling technique that reduces the number of pivots for nodes with high centrality.
Different from the previous approaches is the approximation algorithm
of Riondato and Kornaropoulos~\cite{DBLP:conf/wsdm/RiondatoK14},
which samples a \emph{single} random shortest path at each iteration.
This approach allows theoretical guarantees on the quality of their approximation:
For any two constants $\epsilon,\delta\in(0,1)$, a number of samples
can be defined such that the error on the approximated values is at
most $\epsilon$ with probability at least $1-\delta$.
Because of this guarantee, we use this algorithm as a 
building block of our new approach and refer to it as \textsf{RK}.

\paragraph*{Exact dynamic algorithms.}
Dynamic algorithms update the betweenness values of all nodes
in response to a modification on the graph, which might be an edge
insertion, an edge deletion or a change in an edge's weight. 
% Five algorithms have been proposed so far. 
The first published approach
is \textsf{QUBE} by Lee \etal~\cite{Lee12qube:a}, which relies on the decomposition
of the graph into connected components. When an edge update occurs,
\textsf{QUBE} re-computes the centrality values using \textsf{BA} only
within the affected component. In case the update modifies the
decomposition, this must be recomputed, and new centralities must
be calculated for all affected components. The approach proposed
by Green \etal~\cite{DBLP:conf/socialcom/GreenMB12} for unweighted graphs 
maintains a structure with the previously calculated BC values and
additional information, like the distance
of each node $v$ from every source $s\in V$ and the list of \textit{predecessors}, i.e.\ the nodes
immediately preceding $v$ in all shortest paths from $s$ to $v$. Using this information,
the algorithm tries to limit the recomputations to the nodes whose
betweenness has actually been affected. Kourtellis
\etal~\cite{DBLP:journals/corr/KourtellisMB14} modify the
approach by Green \etal~\cite{DBLP:conf/socialcom/GreenMB12} in
order to reduce the memory requirements. Instead of
storing the predecessors of each node $v$ from each possible
source, they recompute them every time the information is required. 
Kas \etal~\cite{DBLP:conf/asunam/KasWCC13} extend an existing algorithm for
the dynamic APSP problem by Ramalingam and Reps~\cite{Ramalingam92anincremental}
to also update BC scores.
The recent work by Nasre 
\etal~\cite{DBLP:journals/corr/NasrePR13} contains the first dynamic
algorithm for BC (\textsf{NPR}) which is asymptotically faster than recomputing 
from scratch on certain inputs. In particular, when only edge insertions are allowed and
the considered graph is sparse and weighted, their algorithm
takes $O(n^2)$ operations, whereas \textsf{BA} requires
$O(n^2\log n)$ on sparse weighted graphs. Among other things, our 
paper contains the first experimental results obtained by an implementation of this algorithm.

All dynamic algorithms mentioned perform better than
recomputation on certain inputs. Yet, none of them
has a worst-case complexity better than \textsf{BA} on \textit{all graphs} since
all require an update of an APSP problem.
 For this problem, no algorithm exists
which has better worst-case running time than recomputation~\cite{DBLP:journals/algorithmica/RodittyZ11}. In addition, the problem of updating BC
seems even harder than the dynamic APSP problem. Indeed,
the dependencies (and therefore BC) might
need to be updated even on nodes whose distance from the source
has not changed, as they could be part of new shortest paths or not
be part of old shortest paths any more.

\paragraph*{Dynamic SSSP algorithms.} Since our algorithms require the update of 
an SSSP solution, we briefly review also the main results on the incremental SSSP problem 
(i.e.\, the problem of updating distances from a source node after an edge insertion or a batch of edge insertions). 
The most well-known algorithms for single-edge updates are 
\textsf{RR}~\cite{Ramalingam96onthe} by Ramalingam and Reps,
 and \textsf{FMN}~\cite{Frigioni_fullydynamic} by Frigioni \etal
The two approaches are implemented and compared in~\cite{Frigioni97experimentalanalysis},
 showing that \textsf{RR} is the faster one.
For the batch problem, the authors of ~\cite{Ramalingam92anincremental}
propose a novel algorithm named \textsf{SWSF-FP}.
 In~\cite{Frigioni_semi-dynamicalgorithms}
a batch variant of \textsf{FMN} is presented. The experimental analysis by Bauer
 and Wagner~\cite{DBLP:conf/wea/BauerW09} reports the results of the
two algorithms and of some tuned variants of \textsf{SWSF-FP}. The recent study by D'Andrea
 \etal~\cite{DBLP:conf/wea/DAndreaDFLP14} compares the
performances of the algorithms considered in~\cite{DBLP:conf/wea/BauerW09} 
with those of a novel approach~\cite{conf/sirocco/DAndreaDFLP13}.
In particular, for batches of only edge insertions, the tuned variant \textsf{T-SWSF} 
has shown significantly better results than other published algorithms.
Therefore we use \textsf{T-SWSF} as a building block of our incremental BC approximation algorithm.

%\newpage

\section{\textsf{RK} algorithm}
\label{sec:rk}

\begin{wrapfigure}{r}{0.35\textwidth}
  \begin{center}
  \vspace{-8ex}
    \includegraphics[width=0.3\textwidth]{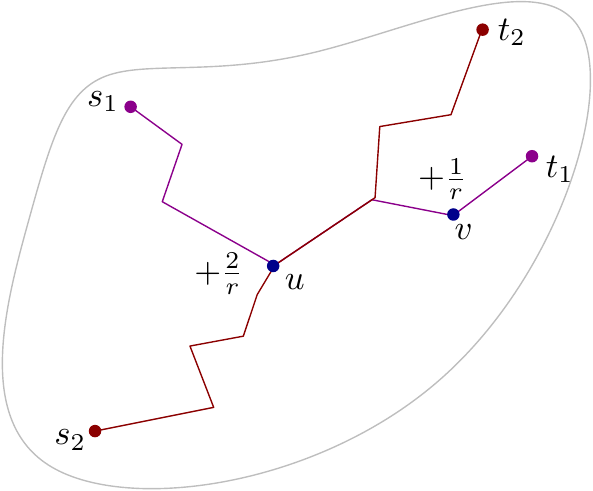}
  \end{center}
  \caption{Sampled paths and score update in the \textsf{RK} algorithm}
  \label{fig:rk-idea}
\end{wrapfigure}
In this section we briefly describe the algorithm by Riondato and Kornaropoulos \textsf{(RK)}~\cite{DBLP:conf/wsdm/RiondatoK14}, 
the foundation for our incremental approach. The idea of \textsf{RK} is to sample a set $S =\{p_{(1)},...,p_{(r)}\}$ of $r$ shortest paths between randomly sampled source-target pairs $(s, t)$.
Then, \textsf{RK} computes the approximated
betweenness centrality $\tilde{c}_B(v)$ of a node $v$ as the
fraction of sampled paths $p_{(k)}\in S$ that $v$ is internal
to, by adding $\frac{1}{r}$ to the node's score for each of these paths.
Figure~\ref{fig:rk-idea} illustrates an example where the sampling of two shortest paths leads to $\frac{2}{r}$ and $\frac{1}{r}$ being added to the score of $u$ and $v$, respectively.
Each possible shortest path $p_{st}$ has the following probability of being sampled in each of the $r$ iterations:
\begin{equation}
\pi_{G}(p_{st})=\frac{1}{n(n-1)}\cdot\frac{1}{\sigma_{st}}
\end{equation}

\noindent The number $r$ of samples required to approximate BC scores with the given error guarantee is calculated as

\begin{equation}
r=\frac{c}{\epsilon^{2}}\left(\lfloor\log_{2}\left(VD(G)-2\right)\rfloor+1+\ln\frac{1}{\delta}\right),
\label{sample_size}
\end{equation}

\noindent where $\epsilon$ and $\delta$ are constants in $(0,1)$, $c \approx 0.5$ and $VD(G)$ is the vertex diameter of $G$, i.e.\ the
number of nodes in the shortest path of $G$ with maximum number of nodes. In unweighted graphs $VD(G)$ coincides with \textsf{diam}$(G)+1$,
where \textsf{diam}$(G)$ is the number of edges in the longest
shortest path. In weighted graphs $VD(G)$ and the (weighted) diameter \textsf{diam}$(G)$ (\ie the length of the longest shortest path)
are unrelated quantities. 
The following error guarantee holds:

\begin{lemma}
\textup{\cite{DBLP:conf/wsdm/RiondatoK14}} \label{lem:1} If $r$ shortest paths
are sampled according to the above-defined probability distribution
$\pi_{G}$, then with probability at least $1-\delta$ the approximations
$\tilde{c}_B(v)$ of the betweenness centralities are within $\epsilon$
from their exact value:
$
\Pr(\exists v\in V\: s.t.\:|c_{B}(v)-\tilde{c}_B(v)|>\epsilon)<\delta.
$

\end{lemma}
To sample the shortest paths according to $\pi_{G}$, \textsf{RK} first chooses
a node pair $(s,t)$ uniformly at random and performs an SSSP search from $s$, keeping also track of the number $\sigma_{sv}$
of shortest paths from $s$ to $v$ and of the list of predecessors
$P_{s}(v)$ for any node $v$. Then one shortest path is selected: 
Starting from $t$, a predecessor $z\in P_{s}(t)$
is selected with probability $\sigma_{sz}/\sum_{w\in P_{s}(t)}\sigma_{sw}=\sigma_{sz}/\sigma_{st}$.
The sampling is repeated iteratively
until node $s$ is reached. Algorithm \ref{algo1} in the Appendix is the
pseudocode for \textsf{RK}. Function \texttt{computeExtendedSSSP} is 
an SSSP algorithm that keeps track of the number of shortest paths and
of the list of predecessors while computing distances, as in \textsf{BA}~\cite{Brandes01betweennessCentrality}.

\paragraph{Approximating the vertex diameter.}

The authors of \textsf{RK} propose two upper bounds on the vertex diameter that can both be computed in $O(n+m)$, instead of solving an APSP problem. For connected unweighted undirected graphs, they compute a SSSP from a randomly-chosen node $s$ and approximate $VD$ as the sum of the two shortest paths starting from $s$ with maximum length. For the remaining graph classes (directed and/or
weighted), the authors approximate $VD$ with the
size of the largest weakly connected component, which, in case of connected graphs,
is equal to the number of nodes, an overestimation for complex networks.
It is important to note that, when the graph is connected and only edge insertions are allowed, the size of both approximations cannot increase as the
graph evolves. This is the key observation on which our method relies.

\section{Incremental approximation algorithms}
\label{sec:dyn-algo}
\subsection{Update after a batch}

\begin{wrapfigure}{r}{0.35\textwidth}
  \begin{center}
  \vspace{-6ex}
    \includegraphics[width=0.25\textwidth]{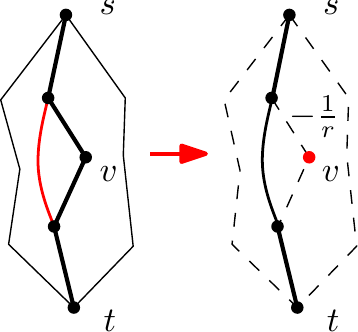}
  \end{center}
  \caption{Updating shortest paths and BC scores}
  \label{fig:path-update}
  \vspace{-4ex}
\end{wrapfigure}
Our incremental algorithms are composed of two phases: an \textit{initialization} phase, which executes \textsf{RK} on the initial graph, and an \textit{update} phase, which recomputes the approximated BC scores after a sequence of edge updates. 
Figure~\ref{fig:path-update} shows an example to illustrate the basic idea: Several shortest $(s,t)$-paths exist, of which one has been sampled. An edge insertion shortens the distance between $s$ and $t$, making the shortest path unique and excluding the node $v$, from whose score $\frac{1}{r}$ must be subtracted. 
From this point on, we give a formal description and only consider edge insertions. We suppose the graph is undirected, but our results can be easily extended to weight decreases and directed graphs. 
A batch $\beta = \{e_1,...,e_k\}$ of edges $e_i=\{u_i,v_i\}$ is inserted into a connected graph
$G$ and we show how the approximated BC scores 
can be updated.
Let $G'=(V,E\cup \beta)$ be the new graph,
let $d'_{s}(t)$ denote the new distance between any node pair $(s,t)$
and let $\sigma'_{st}$ be the new number of shortest paths between $s$ and $t$.
Let $\mathcal{S}_{st}$ and $\mathcal{S}'_{st}$ be the old and the new set of shortest paths between $s$ and $t$, respectively.
A new set $ S'=\{p'_{(1)},...,p'_{(r)}\} $
of shortest paths has to be sampled now in order to let Lemma~\ref{lem:1}
hold for the new configuration; in particular, the probability  $\Pr(p'_{(k)}=p'_{st})$ of
each shortest path $p'_{st}$ to be sampled
must be equal to $\pi_{G'}(p'_{st})=\frac{1}{n(n-1)}\cdot\frac{1}{\sigma'_{st}}$.
Clearly, one could rerun \textsf{RK} on the new graph, but we can be more efficient:
 We recall that the upper bound
on the vertex diameter, and therefore the number $r$ of
samples, cannot increase after the insertion if the graph is connected.
Given any old sampled path $p_{st}$, we can \emph{keep}
$p_{st}$ if the set of shortest paths between $s$ and $t$ has not
changed, and \emph{replace} it with a new path between $s$ and
$t$ otherwise. Then, the following lemma holds:

\begin{lemma}
\label{lem:2}Let $ S$ be a set of shortest paths of $G$
sampled according to $\pi_{G}$. Let $\mathcal{P}$ be the procedure
that creates $ S'$ by substituting each path $p_{st}\in S$
with a path $p'_{st}$ according to the following rules:
\begin{enumerate}
\item $p'_{st}=p_{st}$ if $d'_{s}(t)=d_{s}(t)$ and $\sigma'_{st}=\sigma_{st}$
\item $p'_{st}$ selected uniformly at random among $\mathcal{S}'_{st}$
\mbox{otherwise}.
\end{enumerate}
Then, $p'_{st}$ is a shortest path of $G'$ and the probability of
any shortest path $p'_{xy}$ of $G'$ to be sampled at each iteration is
$\pi_{G'}(p'_{xy})$, i.e.\
\[
\Pr(p'_{(k)}=p'_{xy})=\frac{1}{n(n-1)}\cdot\frac{1}{\sigma'_{xy}},\quad k=1,...,r\text{ .}
\]
\end{lemma}
The proof can be found in Section~\ref{app:proof_lemma} of the Appendix.
Since the set of paths is constructed according to $\pi_{G'}$, Theorem~\ref{theorem_samples} follows directly from Lemma \ref{lem:1}.
\begin{theorem}
\label{thm:approx_guarantee}
Let $G=(V,E)$ be a connected graph and let $G'=(V,E\cup \beta)$
be the modified graph after the insertion of the batch $\beta$. 
%Let $VD(G)\geq VD(G')$.
Let $ S$ be a set of $r$ shortest paths of $G$
sampled according to $\pi_{G}$
and $r=\frac{c}{\epsilon^{2}}\left(\lfloor\log_{2}\left(VD(G)-2\right)\rfloor+1+\ln\frac{1}{\delta}\right)$
for some constants $\epsilon,\delta\in(0,1)$. Then, if a new set
$ S'$ of shortest paths of $G'$ is built according to procedure
$\mathcal{P}$ and the approximated values of betweenness centrality
$\tilde{c}'_B(v)$ of each node $v$ are computed as the fraction
of paths of $ S'$ that $v$ is internal to, then
\[
\Pr(\exists v\in V\: s.t.\:|c'_{B}(v)-\tilde{c}'_B(v)|>\epsilon)<\delta,
\]
where $c'_{B}(v)$ is the new exact value of betweenness centrality
of $v$ after the edge insertion.
\label{theorem_samples}
\end{theorem}
Algorithm~\ref{algo3}
shows the update procedure based on
Theorem~\ref{thm:approx_guarantee}. For each sampled node pair $(s_{i},t_{i}),\, i=1,...,r$,
we first update the SSSP DAG, a step which will be
discussed in Section~\ref{sec:Structure-updates}.
 In case the distance
or the number of shortest paths between $s_{i}$ and $t_{i}$ has
changed, a new shortest path is sampled uniformly as in \textsf{RK}. 
This means that $\frac{1}{r}$ is subtracted from the score of each node in the old shortest path and the same quantity
is added to the nodes in the new shortest
path. On the other hand, if both distances and number of shortest
paths between $s_{i}$ and $t_{i}$ have not changed, nothing needs
to be updated. 

Considering edges in a batch allows us to recompute the BC scores only once instead of doing it after each single edge update. Moreover, this gives us the possibility to use specific batch algorithms for the update of the SSSP DAGs, which process the nodes affected by multiple edges of $\beta$ only once, instead of for each single edge.

\begin{algorithm}[htb]
\begin{small}
\LinesNumbered
\SetKwData{B}{$\tilde{c}_B$}
\SetKwInOut{Input}{Input}\SetKwInOut{Output}{Output}
\Input{Graph $G=(V,E)$, source node $s$, number of iterations $r$, batch $\beta$}
\Output{New approximated BC values of nodes in $V$}
\For{$i \leftarrow 1$ \KwTo $r$}
{
	$d^{old}_i\leftarrow d_{s_i}(t_i)$\;
	$\sigma^{old}_i\leftarrow \sigma_{s_i}(t_i)$\;
	$(d_{s_i},\sigma_{s_i},P_{s_i})\leftarrow$\textsf{UpdateSSSP}($G,d_{s_i},\sigma_{s_i},P_{s_i},\beta$)\; \label{lst:line:update_sssp}
	\tcp{If the shortest paths between $s_i$ and $t_i$ have changed, we sample a new one}
	\If{$d_{s_i}(t_i)<d^{old}_i$ \textbf{or} $\sigma_{s_i}(t_i)\neq \sigma^{old}_i$}
	{
			\tcp{We subtract $1/r$ from all nodes in the old shortest path and we add $1/r$ to all nodes in the new shortest path}
			\ForEach{$w \in p_{(i)}$}
			{
				\B($w$) $\leftarrow \B(w)-1/r$\;
			}
			$v \leftarrow t_i$\;
			$p_{(i)}\leftarrow$ empty list\;
			\While{$P_{s_i}(v) \neq \{s_i\}$}
			{
				\mbox{sample $z \in P_{s_i}(v)$ with $\Pr=\sigma_{s_i}(z)/\sigma_{s_i}(v)$}\;
				$\B(z)\leftarrow \B(z)+1/r$\;
				add $z$ to $p_{(i)}$\;
				$v\leftarrow z$\;
			}
	}
}
\Return{$\{(v,\B(v)),\: v\in V\}$}
\caption{BC update after a batch $\beta$ of edge insertions}\label{algo3}
\end{small}
\end{algorithm}

\vspace{-3ex}

\subsection{SSSP Update Algorithms}
\label{sec:Structure-updates}
In Line~\ref{lst:line:update_sssp} of Algorithm~\ref{algo3}, we need an algorithm \textsf{UpdateSSSP} to update the information
about distances, predecessors and number of shortest paths.
This is an extension to the problem of updating a SSSP tree,
for which some approaches have been proposed in the literature (see Section~\ref{sec:related_work}).
For weighted graphs, we build our \textsf{UpdateSSSP} algorithm upon \textsf{T-SWSF}~\cite{DBLP:conf/wea/BauerW09}; given that no existing 
algorithm for the dynamic SSSP is asymptotically better than recomputing
from scratch, we base our choice on the results of the experimental 
studies~\cite{DBLP:conf/wea/BauerW09,DBLP:conf/wea/DAndreaDFLP14}. 
The difference between our \textsf{UpdateSSSP} for weighted graphs (\textsf{UpdateSSSPW}) and \textsf{T-SWSF} is that
we also need to update the list of predecessors and the number of shortest paths in addition to the distances.
For unweighted graphs, our \textsf{UpdateSSSP} is a \textit{novel} batch SSSP update approach which is in principle similar to 
\textsf{T-SWSF} but has better time bounds. 
%\hmey{potential for shortening in the remainder of the paragraph}
In particular, naming $||A$|| the sum of the number of nodes \textit{affected} by the batch (\ie nodes whose distance or number of shortest paths
from the source has changed) and the number of their incident edges and naming $\beta$ the size
of the batch, our algorithm for unweighted graphs requires $O(|\beta |+ ||A || + d_{max})$ for an update, where $d_{max}$ is at most equal to the vertex diameter.
The time required by \textsf{UpdateSSSPW} (and \textsf{T-SWSF}) is instead $O(|\beta |\log |\beta | + ||A || \log ||A ||)$.
Our two algorithms \textsf{UpdateSSSPW} and \textsf{UpdateSSSP} lead to two
incremental BC approximation algorithms: \textsf{IAW}, which uses \textsf{UpdateSSSPW}, and \textsf{IA},
which uses \textsf{UpdateSSSP}.

\paragraph*{SSSP update for weighted graphs.}
Now we briefly describe \textsf{UpdateSSSPW} (Algorithm~\ref{algo:t-swsf} in Section~\ref{app:pseudocodes}), which is responsible for updating distances, 
list of predecessors and number of shortest paths in weighted graphs after a batch of edge insertions.
We say we \textit{relax} an edge $\{u,v\}$ when we compare $d(v)$ with $d(u)+\mathsf{w}(\{u,v\})$ and, 
in case $d(v)> d(u)+\mathsf{w}(\{u,v\})$, we insert
$v$ in a priority queue $Q$ with priority $p_Q(v) = d(u)+\mathsf{w}(\{u,v\})$ (and the same swapping $u$ and $v$). 
Our algorithm (and \textsf{T-SWSF}) can be divided into two parts: In the first part 
(Lines~\ref{lst:line:start1} -~\ref{lst:line:end1}), it relaxes all the edges of the batch.
In the second part (Lines~\ref{lst:line:main1} -~\ref{lst:line:main2}), each node $w$ in $Q$ is extracted, 
$d(w)$ is set to $p_Q(w)$ and all the incident edges of $w$ are relaxed. 
The algorithm continues until $Q$ is empty and therefore no node needs to be updated.

In this high-level structure, our algorithm is equal to \textsf{T-SWSF}. However, our \textsf{UpdateSSSPW} algorithm
needs to recompute also the information about the shortest
paths. This is accomplished during the scan of the incident edges: In Lines~\ref{lst:line:start2} -~\ref{lst:line:end2}, we
recompute the list of predecessors of node $w$ as the set of neighbors whose distance is lower than the distance
of $w$. The number of shortest paths is then recomputed as the sum of the number of shortest paths of the predecessors.
To study the complexity of \textsf{UpdateSSSPW}, we denote by $\beta$ the set of edges of the batch
and by $A$ the set of affected nodes (i.e.\, vertices whose distance or number of shortest paths
from the source has changed as a consequence of the batch). 
Let $|A|$ and $|\beta |$ represent the cardinality of $A$ and $\beta$, respectively. Moreover, let 
$||A ||$ represent the sum of the nodes in $A$ and of the edges that have at least one endpoint in $ A$.
Then, the following theorem holds.
\begin{theorem}
\label{thm:complexity}
The time required by \textsf{UpdateSSSPW} to update the distances, the number of shortest paths and the list of predecessors is $O(|\beta |\log |\beta | + ||A || \log ||A ||)$.
\end{theorem}

\begin{proof}
In Lines~\ref{lst:line:start1} -~\ref{lst:line:end1} of Algorithm~\ref{algo:t-swsf}, we scan all $| \beta |$ elements of the batch and perform at most
one heap operation (insertion or decrease-key) for each of them, consequently requiring $O(|\beta |\log | \beta |)$ time.
In the main loop (Lines~\ref{lst:line:main1} -~\ref{lst:line:main2}), each affected node is extracted from $Q$ exactly once and its
neighbors are scanned (and possibly inserted/updated within $Q$). This leads to a $O(||A || \log ||A ||)$ complexity for the main loop.
The total time required by the algorithm is therefore $O(|\beta |\log |\beta | + ||A || \log ||A ||)$. 
\end{proof}

%
%\hmey{move proof to appendix if space is needed}

\paragraph*{A new SSSP update algorithm for unweighted graphs.}
In case of unweighted graphs, we present a new algorithm (Algorithm~\ref{algo:unweighted} of Section~\ref{app:pseudocodes}) that can
update the data structures with less computational effort than \textsf{updateSSSPW}.
Since in unweighted graphs the
distances from the source can only be discrete, we can see these distances
as \emph{levels}, which range from $0$ (the source)
to the distance $d_{max}$ of the furthest node from the source.
Therefore, we can replace the priority queue $Q$ of Algorithm~\ref{algo:t-swsf} 
with a list of FIFO queues containing, for each level $k$, the affected nodes
that would have priority $k$ in $Q$. As in Algorithm~\ref{algo:t-swsf}, we
first scan all the edges of the batch (Lines~\ref{lst2:line:start1} -~\ref{lst2:line:end1}), inserting the affected nodes
in the queues. Then (Lines~\ref{lst2:line:start2} -~\ref{lst2:line:end2}), we scan the queues in order of increasing 
distance from the source, in a similar way to the extraction of nodes
from the priority queue. Also in this case, we scan the incident edges to find the new predecessors
(and therefore update the list of predecessors and the number of shortest paths) (Lines~\ref{lst2:line:pred1} -~\ref{lst2:line:pred2})
and to detect other possibly affected nodes (Lines~\ref{lst2:line:succ1} -~\ref{lst2:line:succ2}).
In order not to insert a node in the queues multiple times, 
 we use colors. At the beginning, all the nodes are white;
then, the first time a node is scanned and inserted into a queue,
we set its color to gray, meaning that the node should not be inserted
into a queue any more. However, the target nodes of the batch might
need to be inserted in a queue more than once. Indeed, it is possible
that we initially insert node $v$ at level $d_{s}(u)+1$, but then
we find a shorter path during the main loop (or afterwards in the batch). 
For this reason, we also define another color, black, meaning that the node 
should not be processed any more, even if it is found
again in a priority queue.
Using the same notation introduced with \textsf{UpdateSSSP} for weighted graphs, the following
theorem holds.
\begin{theorem}
\label{thm:complexity2}
The time required by \textsf{UpdateSSSP} to update the distances, the number of shortest paths and the list of predecessors is $O(|\beta |+ ||A ||+d_{max})$.
\end{theorem}
\begin{proof}
The complexity of the initialization step (Lines~\ref{lst2:line:start1} -~\ref{lst2:line:end1}) of Algorithm~\ref{algo:unweighted}
is $O(d_{max}+|\beta|)$, as we initialize a vector of empty 
lists of size $d_{max}$ and scan the batch. In the main loop, we scan again the list
of queues of size $d_{max}$ and, for every node $w$ in one of
the queues, we scan all the incident edges of
$w$. Therefore, the complexity of the main loop is the sum of
the number of nodes in each queue plus the number of edges that
have one endpoint in one of the queues. Using the coloring, each affected
node which is not the target of an edge of the batch is inserted
in $Q$ exactly once. The affected target nodes $v$, instead,
can be inserted at most $n_{v}+1$ times, where $n_{v}$ is the number of
edges in $\beta$ whose target is $v$. The reason of the $+1$ is
that $v$ can be inserted in a queue at most once during the main
loop, as then $color[v]$ will be set to gray. The complexity of the
algorithm is therefore $O(|\beta|+\|A\|+d_{max})$.
\end{proof}

%\hmey{move proof to appendix if space is needed}

\section{Experiments}
\label{sec:experimental}

\subsection{Experimental setup}
\label{sec:exp_setup}
\paragraph{Implementation and settings.} For an experimental comparison, we implemented our two incremental approaches \textsf{IA} and \textsf{IAW}, as well as the static approximation \textsf{RK}, the static exact \textsf{BA}, the dynamic exact algorithms \textsf{GMB}~\cite{DBLP:conf/socialcom/GreenMB12} and \textsf{KDB}~\cite{DBLP:journals/corr/KourtellisMB14} for unweighted graphs and the dynamic exact algorithm \textsf{NPR}~\cite{DBLP:journals/corr/NasrePR13} for weighted graphs. We chose to implement \textsf{NPR} because it is the only algorithm with a lower complexity than recomputing from scratch and no experimental results have been provided before. In addition, we also implemented \textsf{GMB} and \textsf{KDB} because they are the ones with the best speedups on unweighted graphs.
We implemented all algorithms in C++, building on the open-source \textit{NetworKit} framework~\cite{DBLP:journals/corr/StaudtSM14}.
In all experiments we fix $\delta$ to 0.1 while the error bound $\epsilon$ varies.
The machine used has 2 x 8 Intel(R) Xeon(R) E5-2680 cores at 2.7 GHz, of which we use only one, and 256 GB RAM. All computations are sequential to make the comparison to previous work more meaningful.

\paragraph{Data sets.} We use both real-world and synthetic networks. In case of disconnected graphs, we extract the largest connected component first. The properties of real-world networks are summarised in Table~\ref{table:graphs}. We cover a range from relatively small networks, on which we test all the implemented algorithms, to large networks with millions of edges, on which we can execute only the most scalable methods, i.e. \textsf{RK} and our incremental approach. In particular, the introduction of approximation allows us to approach graphs of several orders of magnitude larger than those considered by previous dynamic algorithms~\cite{DBLP:journals/corr/KourtellisMB14}. 
Our test set includes routing/transportation networks in which BC has a straightforward interpretation, as well as social networks in which BC can be understood as a measure of social influence.
Most networks are publicly available from the collection compiled for the 10th DIMACS challenge\footnote{\url{http://www.cc.gatech.edu/dimacs10/downloads.shtml}}~\cite{BaderMSW12dimacs}.
Due to a shortage of actual dynamic network data with only edge dynamics, we simulate dynamics on these real networks by removing a small fraction of random edges (without separating the connected component) and adding them back in batches.
To study the scalability of the methods, we also use synthetic graphs of growing sizes obtained with the Dorogovtsev-Mendes generator, a simple and scalable model for networks with power-law degree distribution~\cite{dorogovtsev2003evolution}. Since \textsf{NPR} and our \textsf{IAW} perform best on weighted graphs, we also generate weighted instances of Dorogovtsev-Mendes networks, by adding random weights to the edges according to a Gaussian distribution with mean 1 and standard deviation 0.1.

\begin{table}
\begin{center}
  \begin{tabular}{ | l | l | r | r | r |}
    \hline

    Graph 					& Type 	& Nodes 		& Edges 		&  Diameter \\ \hline
    \texttt{PGPgiantcompo} 		& social / web of trust 	&10 680		& 24 316		&(24, 24)	\\ 
    \texttt{as-22july06}		 		& power grid 	& 22 963		& 48 436	 	&(11, 12)	\\
     \texttt{caidaRouterLevel}		& internet		& 192 244		  	& 609 066		&(26, 28)		\\
    \texttt{email-Enron}			& social / communication		& 36 692		& 183 831 	&(13, 14)	\\
    \texttt{coAuthorsCiteseer}		& social / coauthorship		& 227 320 	& 814 134		&(33, 36)	\\
    \texttt{coAuthorsDBLP}			& social / coauthorship		& 299 067 	& 977 676		&(24, 26)	\\
    \texttt{citationCiteseer }			& citations		& 268 495 	& 1 156 647	&(36, 38)	\\
    \texttt{belgium.osm}			& street		& 1 441 295	& 1 549 970 	& (1987, 2184)		\\
    \texttt{Texas84.facebook}		& social / friendship		& 	36 371 		& 1 590 655		& (7,7) \\
    \texttt{netherlands.osm}				& street	&  2 216 688	& 2 441 238 	&	(2553, 2808)	\\
    \texttt{coPapersCiteseer} 		& social / coauthorship		& 434 102 	& 16 036 720 	&	(34, 36)	\\
    
     \texttt{coPapersDBLP} 			& social / coauthorship		&540 486	 	& 15 245 729	&	(23, 24)	\\

    \hline
  \end{tabular}
\end{center}
  \caption{Overview of largest components of real graphs used in the experiments.}
  \label{table:graphs}
\end{table}

%\vspace{-8ex}

\subsection{Accuracy}

We consider the accuracy of the approximated centrality scores both in terms of absolute error and, more importantly, the preservation of the ranking order of nodes.
Our incremental algorithm has exactly the same accuracy as \textsf{RK}. 
The authors of~\cite{DBLP:conf/wsdm/RiondatoK14} study the behavior of the algorithm also experimentally, considering the average and maximum estimation error on a small set of real graphs. 
We study the absolute and relative errors (\ie ratio between the approximated and the exact score) on a wider set of graphs. For our experiments we use \texttt{PGPgiantcompo}, \texttt{as-22july06} and Dorogovtsev-Mendes graphs of several sizes. Our results confirm those of~\cite{DBLP:conf/wsdm/RiondatoK14} in the sense that the measured absolute errors are \textit{always} below the guaranteed maximum error $\epsilon$ and the measured average error is orders of magnitude smaller than $\epsilon$. We observe that the relative error is almost constant among the different ranks and increases for high values of $\epsilon$.
We also study the \textit{relative rank error} introduced by Geisberger \etal~\cite{DBLP:conf/alenex/GeisbergerSS08} (\ie $\max \{\rho, 1/\rho\}$, denoting $\rho$ the ratio between the estimated rank and the
true rank), which we consider the most relevant measure of the quality of the approximations. Figure~\ref{fig:rank} shows the results for \texttt{PGPgiantcompo}, a similar trend can be observed on our other test instances as well. On the left, we see the errors for the whole set of nodes (ordered by exact rank) and, on the right, we focus on the top 100 nodes. The plot show that for a small value of $\epsilon$ (0.01), the ranking is very well preserved over all the positions. With higher values of $\epsilon$, the rank error of the nodes with low betweenness increases, as they are harder to approximate. However, the error remains small for the nodes with highest betweenness, the most important ones for many applications.

\begin{figure}[htbp]
\begin{center}
\includegraphics[width = 0.45\textwidth]{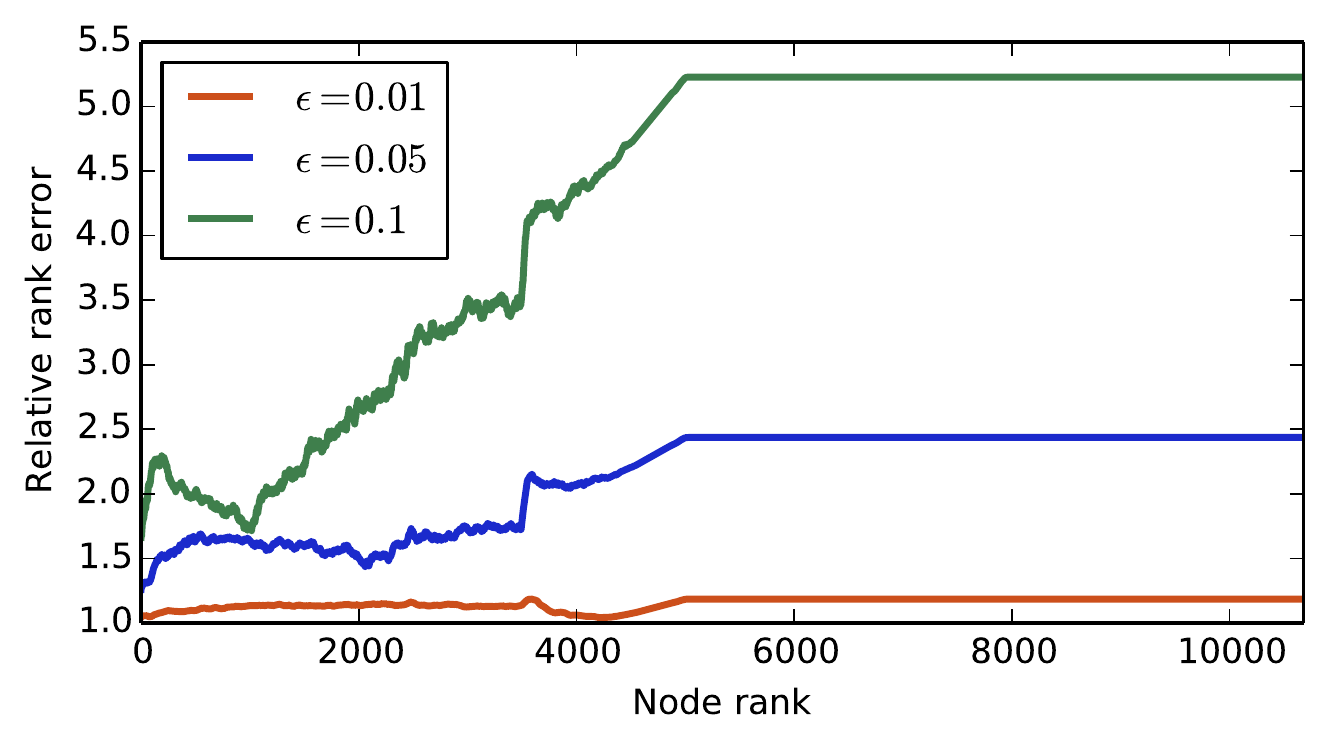}
\includegraphics[width = 0.45\textwidth]{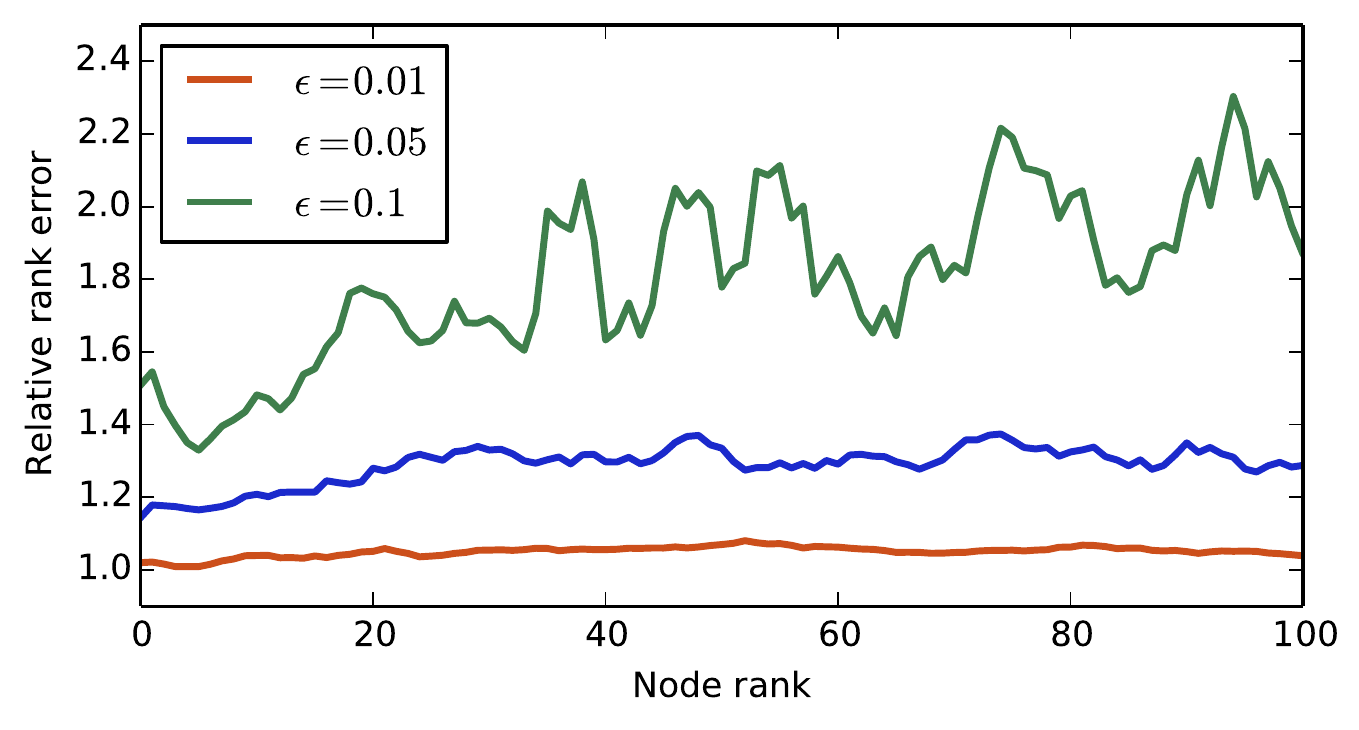}
\caption{Relative rank error on \texttt{PGPgiantcompo} for nodes ordered by rank. On the left, relative errors of all nodes. On the right, relative errors of the 100 nodes with highest betweenness.}
\label{fig:rank}
\end{center}
\end{figure}

\subsection{Running times}
In this section, we discuss the running times of all the algorithms, the speedups of the exact incremental approaches (\textsf{GMB}, \textsf{KDB} and \textsf{NPR}) on \textsf{BA} and the speedups of our algorithm on \textsf{RK}. 
In all of our tests on unweighted graphs, \textsf{KDB} performs worse than \textsf{GMB}, therefore we only report the results of \textsf{GMB}.
 Figure~\ref{pgp} shows the behavior of the four algorithms for unweighted graphs on \texttt{PGPgiantcompo}. Since \textsf{GMB} can only process edges one by one, its running time increases linearly with the batch size, becoming slower than the static algorithm already with a batch size of 64. Our algorithm shows much better speedups and proves to be significantly faster than recomputation even with a batch of size 1024. The reasons for our high speedup are mainly two: First, we process the updates in a batch, processing only once the nodes affected by multiple edge insertions. Secondly, our algorithm does not need to recompute the dependencies, in contrast to all dynamic algorithms based on \textsf{BA} (\ie all existing dynamic exact algorithms). For each SSSP, the dependencies need to be recomputed not only for nodes whose distance or number of shortest paths from the source has changed after the edge insertion(s), but also for all the intermediate nodes in the old shortest paths, even if their distance and number of shortest paths are not affected. This number is significantly higher, since for every node which changes its distance or increases its number of shortest paths, the dependencies of \textit{all} the nodes in \textit{all} the old shortest paths are affected.

Results on synthetic unweighted graphs are analogous to those shown in Figure~\ref{pgp} and can be found in the Appendix (Figure~\ref{unweighted} in Section~\ref{app:experiments}). We did the same comparison on synthetic weighted graphs (described in Section~\ref{sec:exp_setup}), using \textsf{NPR} instead of \textsf{GMB}. We observed that, even though \textsf{NPR} is theoretically faster than recomputing from scratch, its speedups are very small (see Figure~\ref{weighted} in Section~\ref{app:experiments}). Figure~\ref{speedups_dyn} in Section~\ref{app:experiments} shows the speedups of \textsf{IAW} on \textsf{NPR} and \textsf{IA} on \textsf{GMB}. Both \textsf{NPR} and \textsf{GMB} have very high memory requirements ($O(n^2+mn)$), which makes the algorithms unusable on networks with more than a few thousand edges. The memory requirement is the same also for all other existing dynamic algorithms, with the exception of \textsf{KDB}, which requires $O(n^2)$ memory, impractical for large networks. 
Since the exact algorithms are not scalable, for the comparison on larger networks we used only \textsf{RK} and our algorithms. Figure~\ref{speedups} (left) shows the speedups of our algorithm \textsf{IA} on \textsf{coPapersCiteseer}. In this case, the average running time for \textsf{RK} is always approximatively 420 seconds, while for \textsf{IA} it ranges from 0.02 seconds (for a single edge update) to 6.2 seconds (for a batch of 1024 edges). Figure~\ref{speedups} (right) summarises the speedups obtained in the networks of Table~\ref{table:graphs}. The red dots represent the average speedup among the networks for a particular batch size and the shaded areas represent the interval where the different speedups lie. The average speedup is between  $10^3$ and $10^4$ for a single-edge update and between $10$ and $10^2$ for a batch size of 1024. Even for graphs with the worst speedups (basically the smallest ones), a batch of this dimension can always be processed faster than recomputation. 

\begin{table}
\begin{center}
  \begin{tabular}{ | l | r | r | r | r | r | r |}
    \hline

    Graph 					& Time \textsf{RK} [s] 	&  Time \textsf{IA} [s] ($b = 1$) 		& Time \textsf{IA} [s] ($b = 1024$)  & Speedup ($b = 1$) & Speedup ($b = 1024$)\\ \hline
    \texttt{PGPgiantcompo} 		& 3.894	& 0.009	&  1.482	& 432.6 & 2.6	\\ 
    \texttt{as-22july06}		 		& 3.437	&  0.020	&  0.520	& 171.8 & 6.6\\
     \texttt{caidaRouterLevel}		& 64.781	& 0.023	& 2.193 & 2816.5& 29.5	\\
    \texttt{email-Enron}			& 82.568	& 0.016	& 0.750 & 5160.5 & 110.0	\\
    \texttt{coAuthorsCiteseer}		& 65.474 	&  0.009	& 1.008 & 7274.8& 64.9	\\
    \texttt{coAuthorsDBLP}			& 109.939 & 0.009	& 1.201 & 12215.4 & 91.5	\\
    \texttt{citationCiteseer }			& 192.609	& 0.030	& 4.229 & 6420.3& 45.5	\\
    \texttt{belgium.osm}			&  666.144&  1.049	&  48.370 & 635.0 & 13.7	\\
    \texttt{Texas84.facebook}		& 344.075 & 0.147	& 21.829  & 2340.6 & 15.7	 \\
    \texttt{netherlands.osm}			& 878.148	&  1.292	&  31.802 & 679.6 & 27.6   \\
    \texttt{coPapersCiteseer} 		& 427.142	& 0.025	& 6.186 & 17085.6 & 69.0	\\    
     \texttt{coPapersDBLP} 			& 497.617	& 0.009	& 3.027 & 55290.7 &	164.3
    \\
    \hline
  \end{tabular}
\end{center}
  \caption{Average running times and speedups on real graphs for batch size $b$. All the algorithms were run with $\epsilon=0.05$ and $\delta=0.1$ and averaged over 10 runs.}
  \label{table:times}
\end{table}

%\vspace{-6ex}

\begin{figure}[htbp]
\begin{center}
\includegraphics[width = 0.45\textwidth]{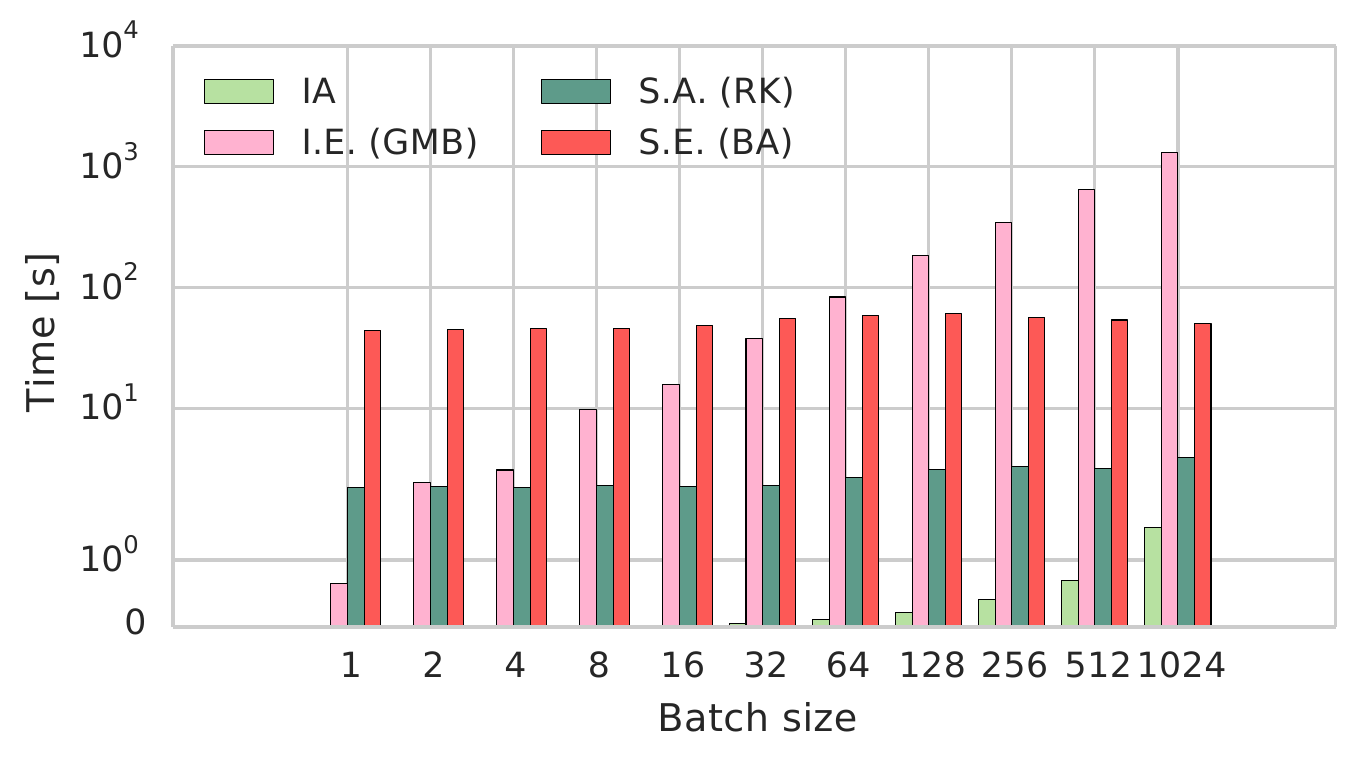}
\includegraphics[width = 0.45\textwidth]{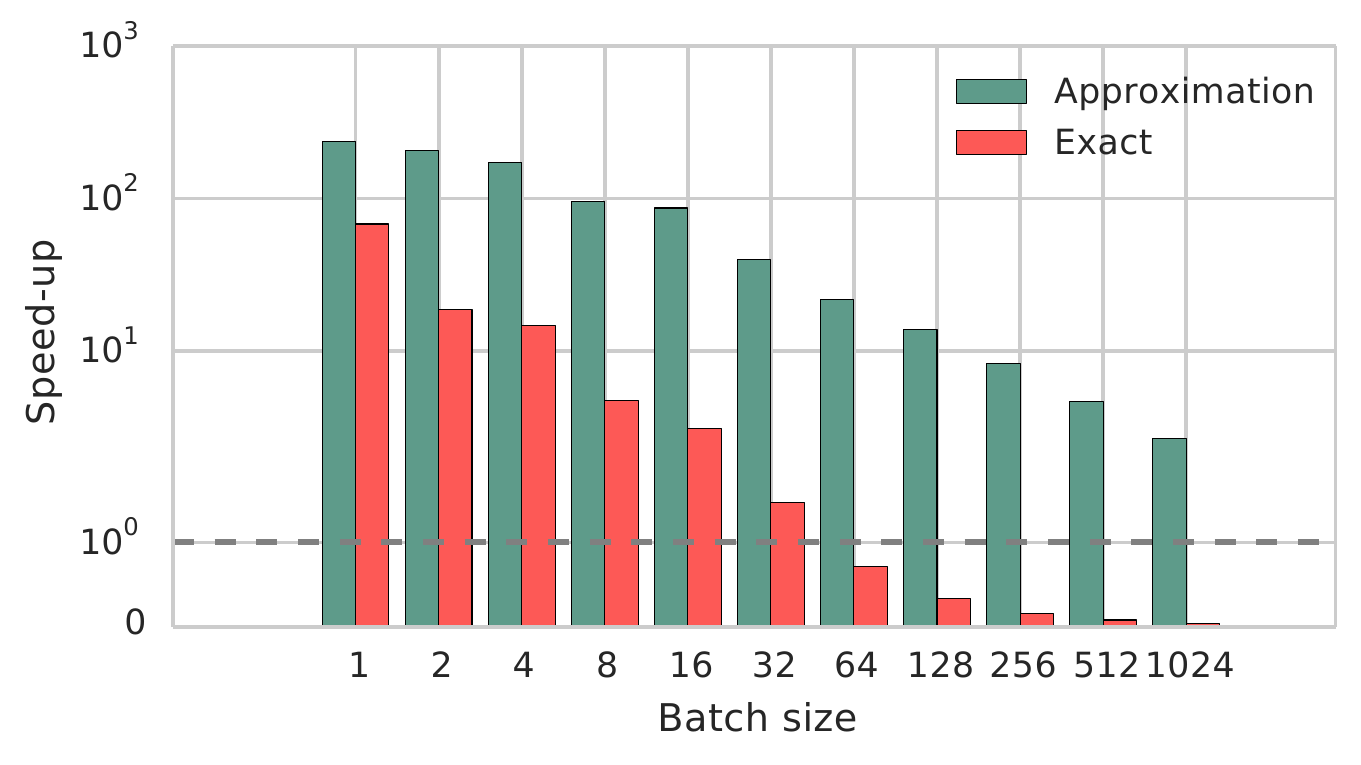}
\caption{Running times and speedups on \texttt{PGPgiantcompo}, with $\epsilon = 0.05$ and with batches of different sizes. On the left, running times of the four algorithms: static exact (\textsf{BA}), static approximation (\textsf{RK}), incremental exact (\textsf{GMB}) and our incremental approximation \textsf{IA}. On the right, comparison of the speedups of \textsf{GMB} on \textsf{BA} and of \textsf{IA} on \textsf{RK}.} 
\label{pgp}
\end{center}
\end{figure}

\begin{figure}[htbp]
\begin{center}
\includegraphics[width = 0.45\textwidth]{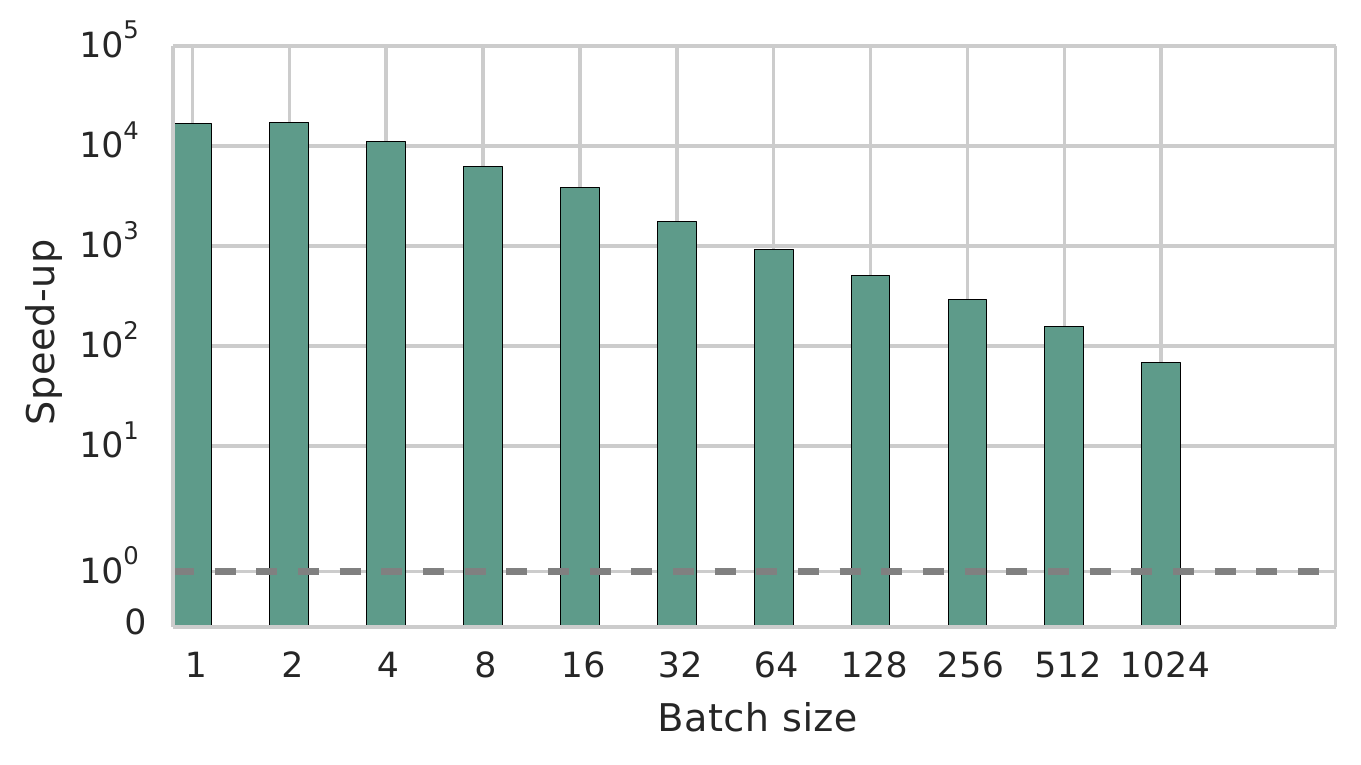}
\includegraphics[width = 0.45\textwidth]{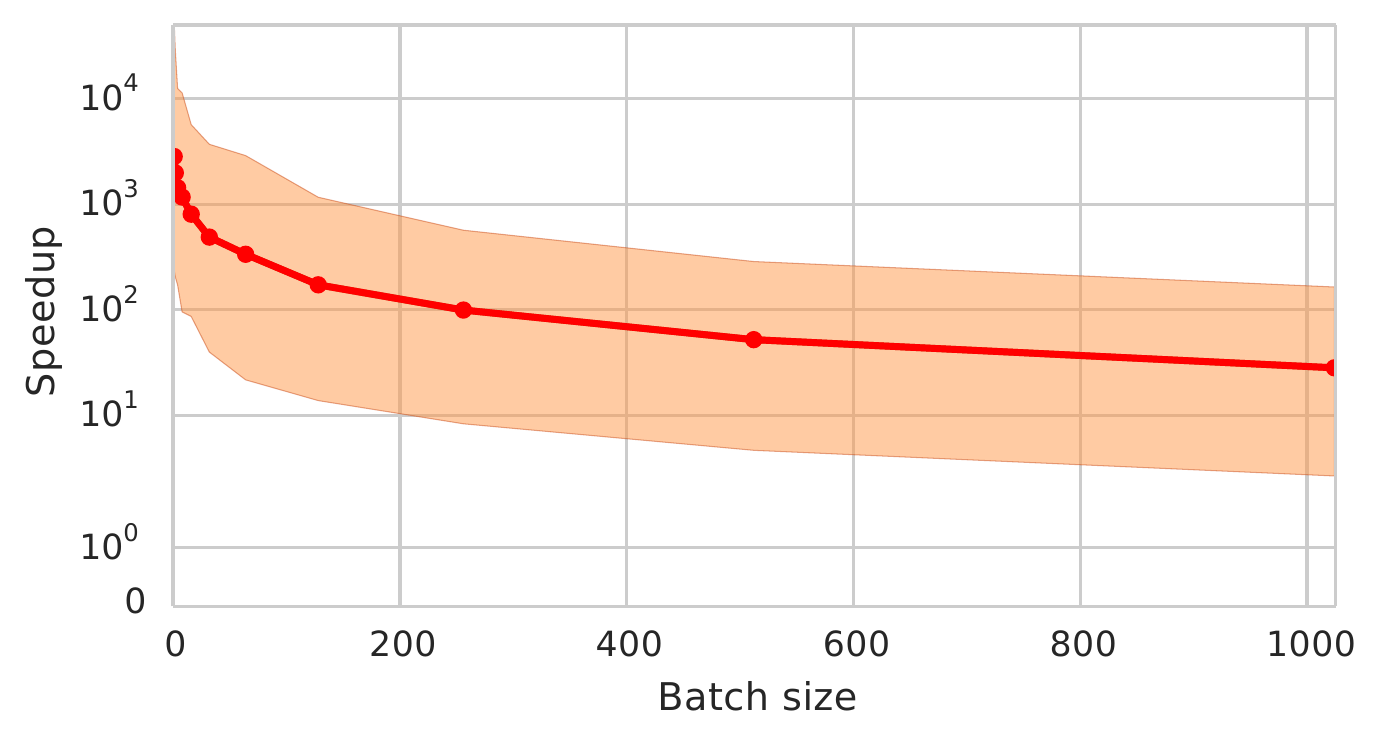}
\caption{On the left, speedups of \textsf{IA} on \textsf{RK} on the \textsf{coPapersCiteseer} network, run with $\epsilon = 0.05$. On the right, geometric mean of speedups over all real networks.}
\label{speedups}
\end{center}
\end{figure}

%\vspace{-6ex}

\section{Conclusions}
Because betweenness centrality considers all shortest paths between pairs of nodes, its exact computation is out of reach for the large complex networks that come up in many applications today. 
However, approximate scores obtained by sampling paths are often sufficient to identify the most important nodes and rank nodes in an order that is very similar to exact BC. 
Since many applications are concerned with rapidly evolving networks, we have explored whether a dynamic approach -- which updates paths when a batch of new edges arrives -- is more efficient than recomputing BC scores from scratch.
Our conclusions are drawn from experiments on a diverse set of real-world networks with simulated dynamics.
We introduce two dynamic betweenness approximation algorithms, for weighted and unweighted graphs respectively.
Their common method is based on the static \textsf{RK} approximation algorithm and inherits its provable bound on the maximum error.
Through reimplementation and experiments, we show that previously proposed dynamic BC algorithms scale badly, especially due to an $O(n^2)$ memory footprint.
Combining a dynamic approach with approximation, our method is more scalable and is the first dynamic one applicable to networks with millions of edges (without using external memory).

As an intermediate result of independent interest, we also propose an asymptotically faster algorithm for updating a solution to the SSSP problem in unweighted graphs.
The dynamic updating of paths implies a higher memory footprint, but also enables significant speedups compared to recomputation (\eg factor 100 for a batch of 1024 edge insertions).
The scalability of our algorithms is primarily limited by the available memory. Each sampled path requires $O(n)$ memory and the number of required samples grows quadratically as the error bound $\epsilon$ is tightened. 
This leaves the user with a tradeoff between running time and memory consumption on the one hand and BC score error on the other hand.
However, our experiments indicate that even a relatively high error bound (\eg $\epsilon = 0.1$) for the BC scores preserves the ranking for the more important nodes well.

We studied sequential implementations for simplicity and comparability with related work, but parallelization is possible and can yield further speedups in practice.
%Further work could lead to an extension of our algorithms that is applicable to unconnected graphs, keeping the error guarantees by possibly sampling new paths as a response to edge updates. \hmey{Difficult, not sure we should mention it...}
Our implementations are based on \textit{NetworKit}\footnotemark, the open-source framework for high-performance network analysis, and we plan to publish our source code in upcoming releases of the package.

\footnotetext{\url{http://networkit.iti.kit.edu}}

\begin{footnotesize}
\paragraph*{Acknowledgements.}
This work is partially supported by German Research Foundation (DFG) grant FINCA (\emph{Fast Inexact Combinatorial and Algebraic Solvers for Massive Networks}) within the Priority Programme 1736 \emph{Algorithms for Big Data} and project \emph{Parallel Analysis of Dynamic Networks -- Algorithm Engineering of Efficient Combinatorial and Numerical Methods}, which is funded by the Ministry of Science, Research and the Arts Baden-W\"{u}rttemberg.
We thank Pratistha Bhattarai (Smith College) for help with the experimental evaluation. We also thank numerous contributors to the \textit{NetworKit} project.
\end{footnotesize}

\newpage

%%%%% bib %%%%%%%%%
\bibliographystyle{plain}
\bibliography{references}
%%%%%%%%%%%%%%%%

%%%%% appendix %%%%%%
\newpage
\appendix
\section{Proof of Lemma~\ref{lem:2}}
\label{app:proof_lemma}
\begin{proof}
To see that $p'_{st}$ is a shortest path of $G'$, it is sufficient
to notice that, if $d'_{s}(t)=d_{s}(t)$ and $\sigma'_{s}(t)=\sigma_{s}(t)$,
then all the shortest paths between $s$ and $t$ in $G$ are shortest
paths also in $G'$.

Let $p'_{xy}$ be a shortest path of $G'$ between nodes $x$ and
$y$. Basically, there are two possibilities for $p'_{xy}$ to be the $k$-th sample. Naming $e_{1}$ the event $\{\mathcal{S}_{xy}=\mathcal{S}'_{xy}\}$
(the set of shortest paths between $x$ and $y$ does not change after
the edge insertion) and $e_{2}$ the complementary event of $e_{1}$,
we can write $\Pr(p'_{(k)}=p'_{xy})$ as $\Pr(p'_{(k)}=p'_{xy}\cap e_{1})+\Pr(p'_{(k)}=p'_{xy}\cap e_{2})$.

Using conditional probability, the first addend can be rewritten as
$\Pr(p'_{(k)}=p'_{xy}\cap e_{1})=\Pr(p'_{(k)}=p'_{xy}|e_{1})\Pr(e_{1})$.
As the procedure $\mathcal{P}$ keeps the old shortest path when $e_{1}$
occurs, then $\Pr(p'_{(k)}=p'_{xy}|e_{1})=\Pr(p_{(k)}=p'_{xy}|e_{1})=\frac{1}{n(n-1)}\frac{1}{\sigma_{x}(y)}$,
which is also equal to $\frac{1}{n(n-1)}\frac{1}{\sigma'_{x}(y)}$,
since $\sigma_{x}(y)=\sigma'_{x}(y)$ when we condition on $e_{1}$.
Therefore, $\Pr(p'_{(k)}=p'_{xy}\cap e_{1})=\frac{1}{n(n-1)}\frac{1}{\sigma'_{x}(y)}\cdot\Pr(e_{1})$.

Analogously, $\Pr(p'_{(k)}=p'_{xy}\cap e_{2})=\Pr(p'_{(k)}=p'_{xy}|e_{2})\Pr(e_{2})$.
In this case, $\Pr(p'_{(k)}=p'_{xy}|e_{2})=\frac{1}{n(n-1)}\cdot\frac{1}{\sigma'_{x}(y)}$,
since this is the probability of the node pair $(x,y)$ to be the $k$-th sample in the initial sampling and of $p'_{xy}$ to be selected
among other paths in $\mathcal{S}'_{xy}$. Then, $\Pr(p'_{(k)}=p'_{xy}\cap e_{2})=\frac{1}{n(n-1)}\cdot\frac{1}{\sigma'_{x}(y)}\cdot\Pr(e_{2})=\frac{1}{n(n-1)}\cdot\frac{1}{\sigma'_{x}(y)}\cdot(1-\Pr(e_{1}))$.
The probability $\Pr(p'_{(k)}=p'_{xy})$ can therefore be
rewritten as 
\[ 
\Pr(p'_{(k)}=p'_{xy}) =\mbox{\ensuremath{\frac{1}{n(n-1)}\frac{1}{\sigma'_{x}(y)}\cdot\Pr}(\ensuremath{e_{1}})+\ensuremath{\frac{1}{n(n-1)}\frac{1}{\sigma'_{x}(y)}\cdot}(1-\ensuremath{\Pr}(\ensuremath{e_{1}}))}=\frac{1}{n(n-1)}\frac{1}{\sigma'_{x}(y)}.
\]
\qed 
\end{proof}

\section{Pseudocode of \textsf{RK}, \textsf{UpdateSSSP} and \textsf{UpdateSSSPW}.}
\label{app:pseudocodes}
\begin{algorithm}[h]
 \begin{small}
\LinesNumbered
\SetKwData{B}{$\tilde{c}_B$}\SetKwData{VD}{VD}
\SetKwFunction{getVertexDiameter}{getVertexDiameter}
\SetKwFunction{sampleUniformNodePair}{sampleUniformNodePair}
\SetKwFunction{computeExtendedSSSP}{computeExtendedSSSP}
\SetKwInOut{Input}{Input}\SetKwInOut{Output}{Output}
\Input{Graph $G=(V,E),\epsilon,\delta \in (0,1)$}
\Output{A set of approximations of the betweenness centrality of the nodes in $V$}
\ForEach{node $v \in V$}
{
	\B$(v)\leftarrow 0$\;
}
\VD($G$)$\leftarrow$\getVertexDiameter{$G$}\;
$r \leftarrow (c/\epsilon^2) (\lfloor \log_2(\VD(G)-2)\rfloor +\ln(1/\delta))$\;
\For{$i \leftarrow 1$ \KwTo $r$}{
	$(s_i,t_i)\leftarrow$ \sampleUniformNodePair{$V$}\;
	\mbox{$(d_{s_i},\sigma_{s_i},P_{s_i})\leftarrow$ \computeExtendedSSSP{$G,s_i$}}\;
	\tcp{Now one path from $s_i$ to $t_i$ is sampled uniformly at random}
	$v \leftarrow t_i$\;
	$p_{(i)}\leftarrow$ empty list\;
	\While{$P_{s_i}(v) \neq \{s_i\}$}
	{
		\mbox{sample $z \in P_{s_i}(v)$ with $\Pr=\sigma_{s_i}(z)/\sigma_{s_i}(v)$}\;
		$\B(z)\leftarrow \B(z)+1/r$\;
		add $z\rightarrow p_{(i)}$;
		$v\leftarrow z$\;
	}
}
\Return{$\{(v,\B(v)),\: v\in V\}$}

\caption{Computation of the approximated betweenness centralities on the initial graph $G$ (\textsf{RK} algorithm)} \label{algo1}
\end{small}
\end{algorithm}

\begin{algorithm}[h]
\begin{small}
\LinesNumbered
\SetKwData{B}{$\tilde{c}_B$}
\SetKwFunction{extractMin}{extractMin}
\SetKwInOut{Input}{Input}\SetKwInOut{Output}{Output}
\Input{Graph $G=(V,E)$, source node $s$, old values of $d_s$, $\sigma_s$ and $P_s$, batch of edge insertion $\beta={e_1,...,e_k}$}
\Output{Updated values of $d_s$, $\sigma_s$ and $P_s$}
$Q \leftarrow$ empty min-based priority queue\;
\ForEach{$e=\{u,v\} \in \beta$}
{\label{lst:line:start1}
	insert $e$ into $E$\;
	\If{$d_s(v)\geq d_s(u)+\mathsf{w}(e)$}
	{
		insert $v$ in $Q$ with $p_Q(v)=d_s(u)+\mathsf{w}(e)$ or, if already in $Q$ and $p_Q(v)>d_s(u)+\mathsf{w}(e)$, set $p_Q(v)\leftarrow d_s(u)+\mathsf{w}(e)$;
	}
	\If{$d_s(u)\geq d_s(v)+\mathsf{w}(e)$}
	{
		insert $u$ in $Q$ with $p_Q(u)=d_s(v)+\mathsf{w}(e)$ or, if already in $Q$ and $p_Q(u)>d_s(v)+\mathsf{w}(e)$, set $p_Q(u)\leftarrow d_s(v)+\mathsf{w}(e)$;
	}
}\label{lst:line:end1}
\While{$Q\neq \emptyset$}
{\label{lst:line:main1}
	$(w,p) \leftarrow$ \extractMin($Q$)\;
	$d_s(w)\leftarrow p$; $P_s(w)\leftarrow \emptyset$; $\sigma_s(w)\leftarrow 0$\;
	\ForEach{incident edge $\{z,w\}$}
	{\label{lst:line:start2}
		\If{$d_s(w)=d_s(z)+\mathsf{w}(\{z,w\})$}
		{
			add $z$ to $P_s(w)$\;
			$\sigma_s(w)\leftarrow \sigma_s(w)+\sigma_s(z)$\;
		}
		\If{$d_s(z)\geq d_s(w)+\mathsf{w}(\{w,z\})$}
			{
				insert $z$ in $Q$ with $p_Q(z)=d_s(w)+\mathsf{w}(\{z,w\})$ or, if already in $Q$ and $p_Q(z)>d_s(w)+\mathsf{w}(\{z,w\})$, set $p_Q(z)\leftarrow d_s(w)+\mathsf{w}(\{z,w\})$;
			}

	}\label{lst:line:end2}
}\label{lst:line:main2}
\Return{$(d_s,\sigma_s,P_s)$}
\caption{\textsf{UpdateSSSPW}}\label{algo:t-swsf}
\end{small}
\end{algorithm}

\begin{algorithm}[h] \begin{small}
\LinesNumbered
\SetKwData{visited}{visited}
\SetKwFunction{extractMin}{extractMin}
\SetKwInOut{Assume}{Assume}
\SetKwInOut{Input}{Input}\SetKwInOut{Output}{Output}
\Input{Graph $G=(V,E)$, source node $s$, old values of $d_s$, $\sigma_s$ and $P_s$, batch of edge insertion $\beta={e_1,...,e_k}$, maximum distance $d_{max}$}
\Output{Updated values of $d_s$, $\sigma_s$ and $P_s$}
\Assume{$color[v]=white$ for all $v\in V$}
$Q[]\leftarrow$ array of empty queues of size $d_{max}$\;
\ForEach{$e=\{u,v\} \in \beta$}
{\label{lst2:line:start1}
	insert $e$ into $E$\;
	\If{$d_s(u)<d_s(v)$}
	{
	$k \leftarrow d_s(u)+1$; enqueue $v \rightarrow Q[k]$\;
	}
	\If{$d_s(v)<d_s(u)$}
	{
	$k \leftarrow d_s(v)+1$; enqueue $u \rightarrow Q[k]$\;
	}
	
}\label{lst2:line:end1}
$k \leftarrow 1$\;
\While{$k<d_{max}$}
{\label{lst2:line:start2}
	\While{$Q[k]\neq \emptyset$}
	{
		dequeue $w\leftarrow Q[k]$\;
		\If{color$(w)=$black}
		{
			continue\;
		}
		$color(w)\leftarrow black$; $d_s(w)\leftarrow k$; $P_s(w)\leftarrow \emptyset$; $\sigma_s(w)\leftarrow 0$\;
		\ForEach{incident edge $(z,w)$}
		{
			\If{$d_s(w)=d_s(z)+1$}
			{\label{lst2:line:pred1}
				add $z$ to $P_s(w)$; $\sigma_s(w)\leftarrow \sigma_s(w)+\sigma_s(z)$\;
			}\label{lst2:line:pred2}
			\If{color$[z]=$white \textbf{and} $d_s(z)\geq d_s(w)+1$}
			{\label{lst2:line:succ1}
				$color[z]\leftarrow gray$; enqueue $z \rightarrow Q[k+1]$\;
			}\label{lst2:line:succ2}
		}

	}
	$k\leftarrow k+1$\;
}\label{lst2:line:end2}
Set to white all the nodes that have been in $Q$\;
\Return{$(d_s,\sigma_s,P_s)$}
\caption{\textsf{UpdateSSSP}}\label{algo:unweighted}
\end{small} \end{algorithm}

\clearpage
\newpage

\section{Additional experimental results}
\label{app:experiments}

%\vspace{-20ex}

\begin{figure}[h!]

\begin{subfigure}[b]{\textwidth}
\begin{center}
\includegraphics[width = 0.45\textwidth]{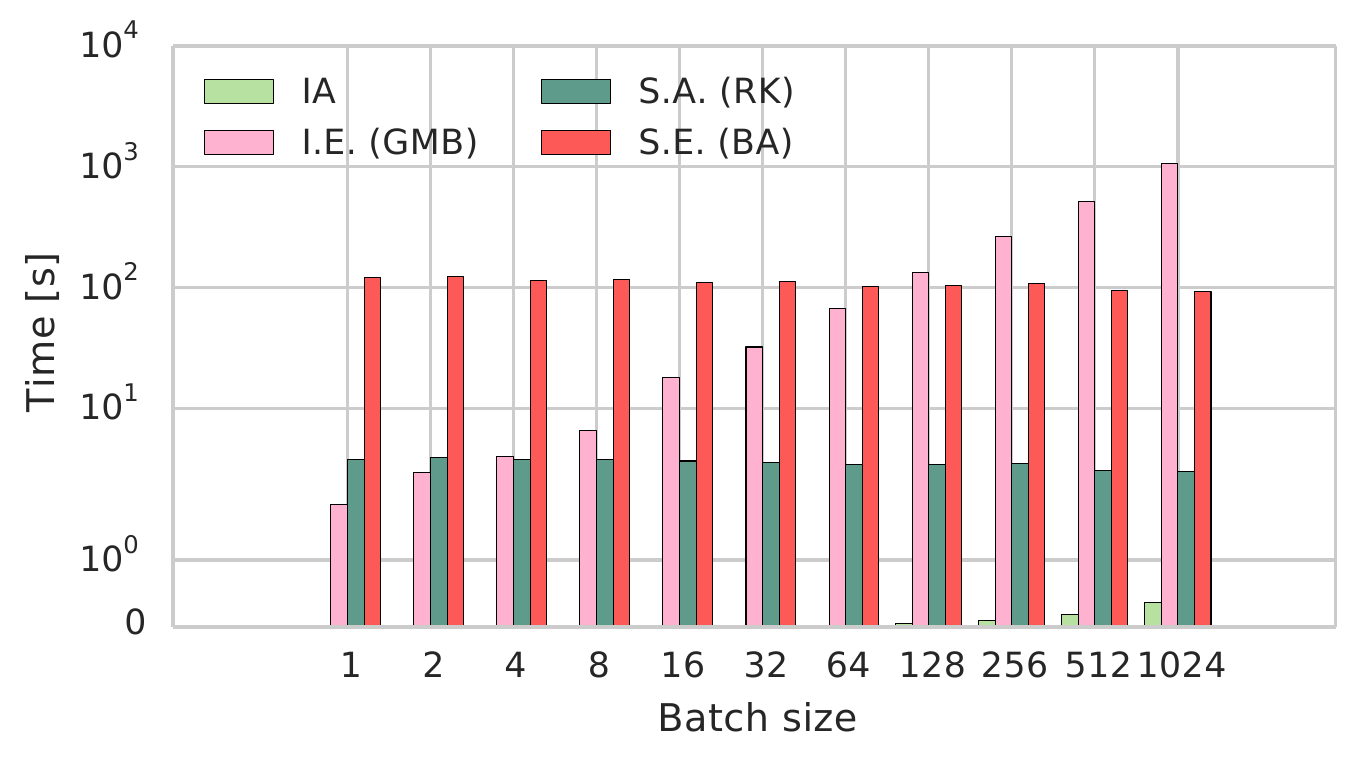}
\includegraphics[width = 0.45\textwidth]{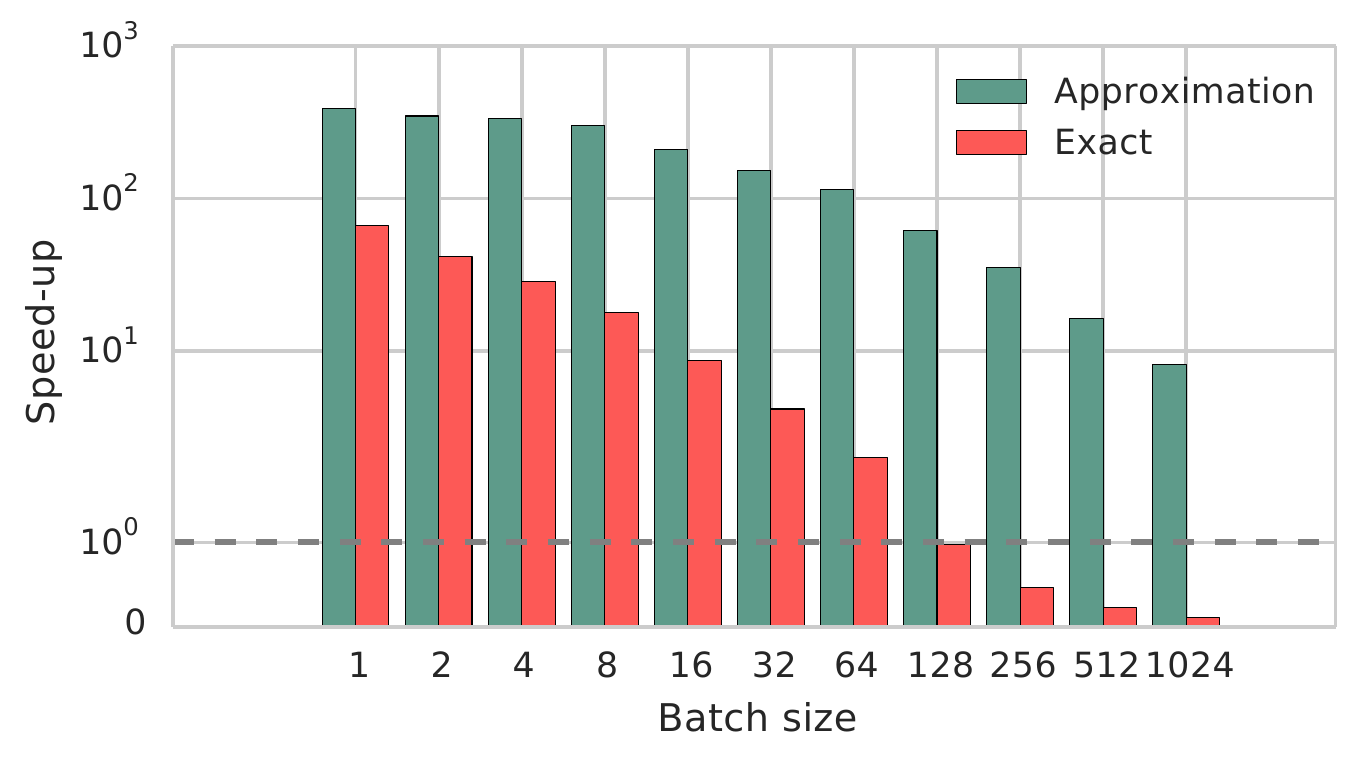}
\caption{Running times and speedups on Dorogovtsev-Mendes synthetic graphs ($m= 40k$), with $\epsilon = 0.05$ with batches of different sizes. Left: running times of the four algorithms: static exact (\textsf{BA}), static approximation (\textsf{RK}), incremental exact (\textsf{GMB}) and \textsf{IA}. Right: comparison of the speedups of \textsf{GMB} on \textsf{BA} and of \textsf{IA} on \textsf{RK}.}
\label{unweighted}
\end{center}
\end{subfigure}

\begin{subfigure}[b]{\textwidth}
\begin{center}
\includegraphics[width = 0.45\textwidth]{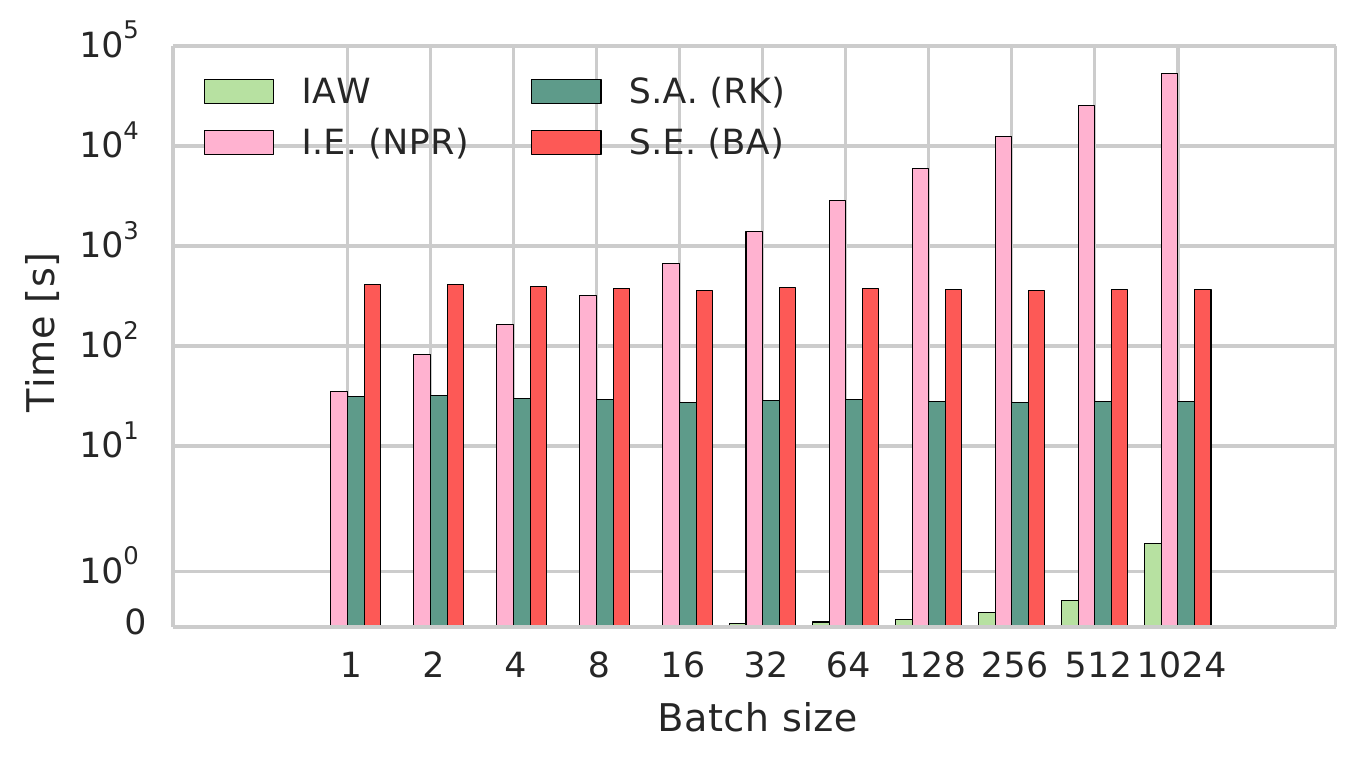}
\includegraphics[width = 0.45\textwidth]{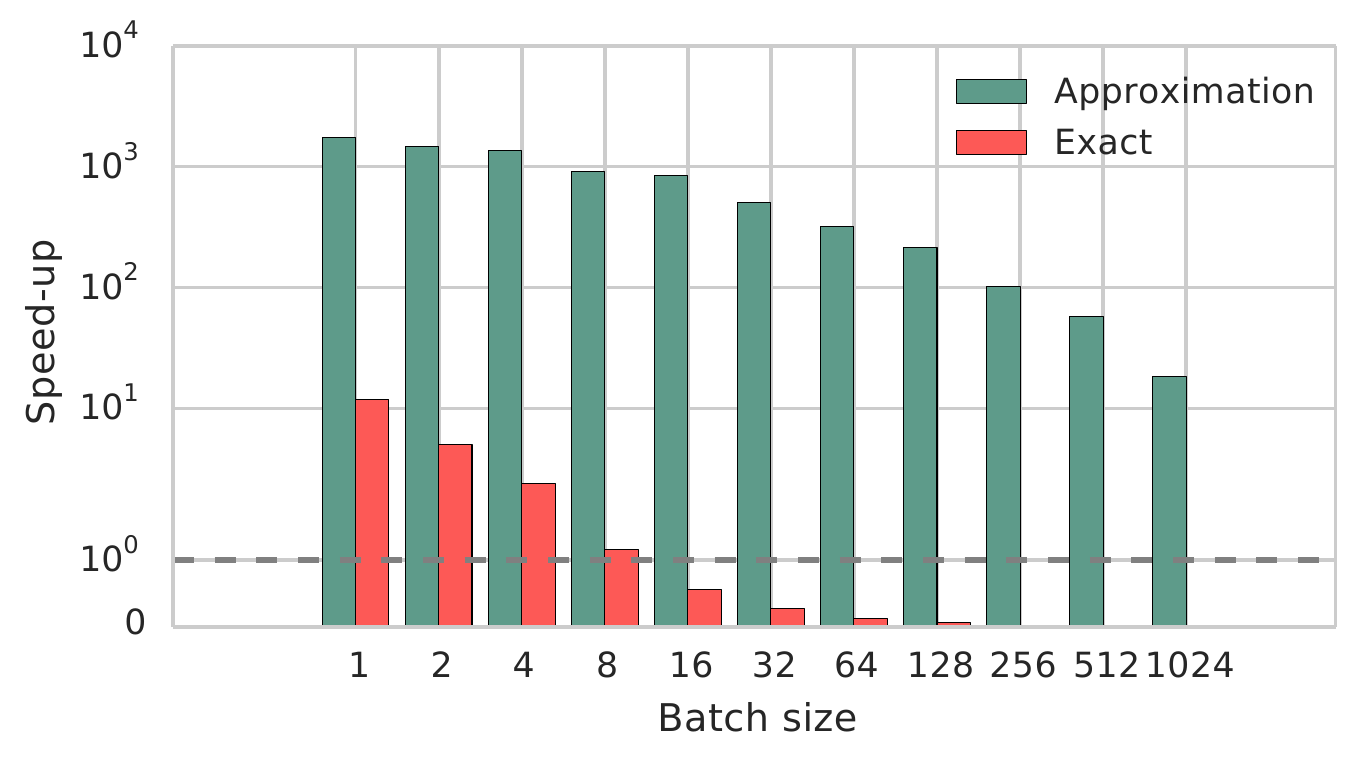}
\caption{Running times and speedups on \textit{weighted} synthetic graphs ($m= 40k$), with $\epsilon = 0.05$ with batches of different sizes. Left: running times of the four algorithms: static exact (\textsf{BA}), static approximation (\textsf{RK}), incremental exact (\textsf{NPR}) and \textsf{IAW}. Right: comparison of the speedups of \textsf{NPR} on \textsf{BA} and of \textsf{IAW} on \textsf{RK}.}
\label{weighted}
\end{center}
\end{subfigure}

\begin{subfigure}[b]{\textwidth}
\begin{center}
\includegraphics[width = 0.45\textwidth]{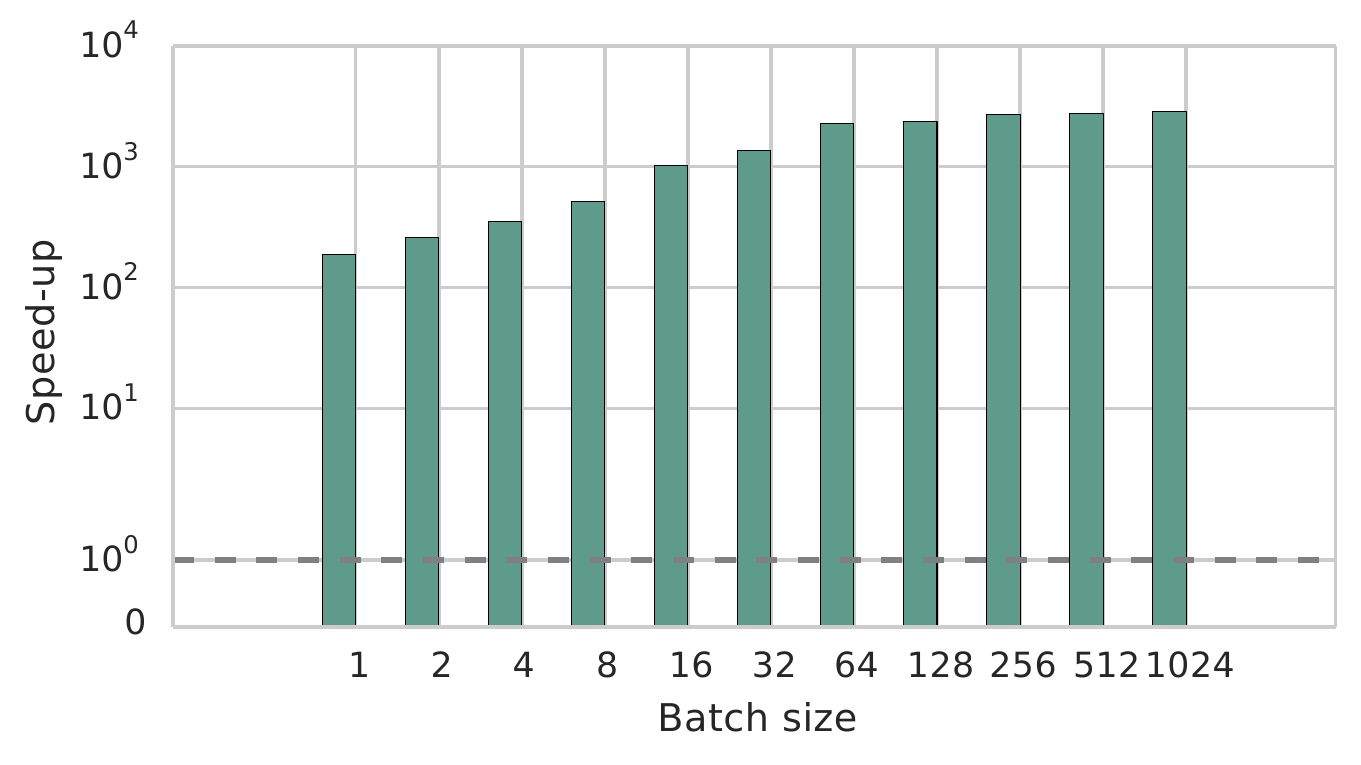}
\includegraphics[width = 0.45\textwidth]{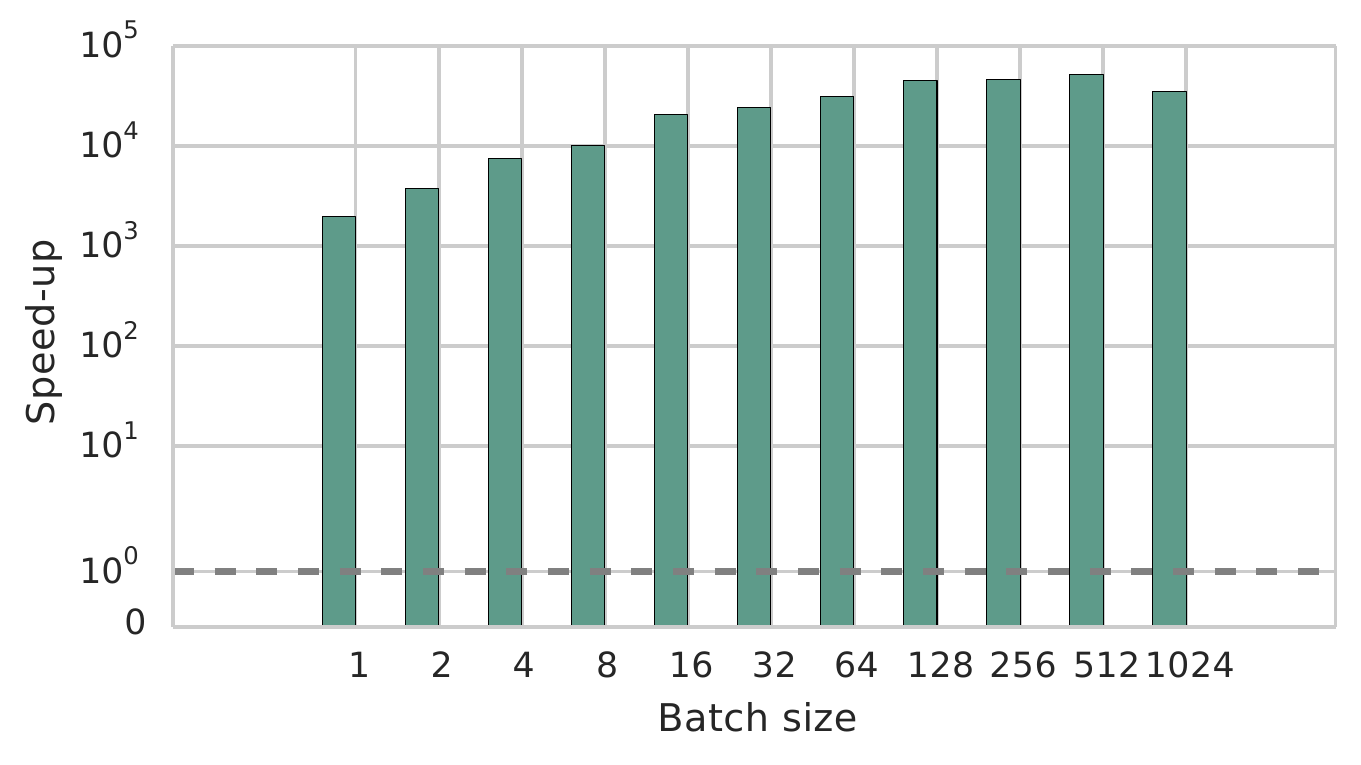}
\caption{Left: speedups of \textsf{IA} on \textsf{GMB} on unweighted synthetic graphs ($m= 40k$). Right: speedups of \textsf{IAW} on \textsf{NPR} on weighted synthetic graphs. \textsf{IA} and \textsf{IAW} run with $\epsilon = 0.05$.}
\label{speedups_dyn}
\end{center}
\end{subfigure}

\caption{Additional experimental results}
\end{figure}

\end{document}